\documentclass[10pt,journal,cspaper,compsoc]{IEEEtran}
\ifCLASSINFOpdf
   \usepackage[pdftex]{graphicx}
\else
   \usepackage[dvips]{graphicx}
\fi

\usepackage{graphicx}
\usepackage{wrapfig}
\usepackage{verbatim}
\usepackage{fancyhdr}
\usepackage{url}
\usepackage{supertabular}
\usepackage{multicol}
\usepackage{setspace}
\usepackage{cite}
\usepackage{subfigure}
\usepackage[small]{caption}
\usepackage{tabularx}
\usepackage{amssymb}
\usepackage{amsmath}
\usepackage{amsbsy}
\usepackage{amsthm}
\usepackage{bbm}

\usepackage{multirow}
\usepackage{stmaryrd}
\usepackage[linesnumbered,ruled]{algorithm2e}
\usepackage{color}

\renewcommand{\textfloatsep}{1mm}
\renewcommand{\dbltextfloatsep}{1mm}
\renewcommand{\belowcaptionskip}{0cm}
\renewcommand{\abovecaptionskip}{0cm}

\renewcommand{\baselinestretch}{1}

\newcolumntype{L}[1]{>{\raggedright\let\newline\\\arraybackslash\hspace{0pt}}m{#1}}
\newcolumntype{C}[1]{>{\centering\let\newline\\\arraybackslash\hspace{0pt}}m{#1}}
\newcolumntype{R}[1]{>{\raggedleft\let\newline\\\arraybackslash\hspace{0pt}}m{#1}}

\ifCLASSOPTIONcompsoc
\else
\fi

\ifCLASSINFOpdf
\else
\fi

\begin{document}
%
\title{Multi-Period Network Rate Allocation with End-to-End Delay Constraints}

\author{Mohammad~H.~Hajiesmaili,
        Ahmad~Khonsari,
        and~Mohammad~Sadegh~Talebi
\IEEEcompsocitemizethanks{\IEEEcompsocthanksitem 		Mohammad H. Hajiesmaili is with the Institute of Network Coding, The Chinese University of Hong Kong, Sha Tin, N.T. Hong Kong,
	{\tt\small mohammad@inc.cuhk.edu.hk} \protect\\
\IEEEcompsocthanksitem Ahmad Khonsari is with School of Electrical and Computer Engineering, College of Engineering of University of Tehran, and with the School of Computer Science, Institute for Research in Fundamental Sciences (IPM), Niavaran Sq., Tajrish Sq., Tehran, IRAN,
{\tt\small ak@ipm.ir}\protect\\
\IEEEcompsocthanksitem M. S. Talebi is with the School of Electrical Engineering, The Royal Institute of Technology (KTH), 100 44, Stockholm, SWEDEN.\protect\\
E-mail: mstms@kth.se}
\thanks{}}


\IEEEcompsoctitleabstractindextext{%
\begin{abstract}
QoS-aware networking applications such as real-time streaming and video surveillance systems require nearly fixed average end-to-end delay over long periods to communicate efficiently, although may tolerate some delay variations in short periods. This variability exhibits complex dynamics that makes rate control of such applications a formidable task. This paper addresses rate allocation for heterogeneous QoS-aware applications that preserves the long-term end-to-end delay constraint while, similar to Dynamic Network Utility Maximization (DNUM), strives to achieve the maximum  network utility aggregated over a fixed time interval. 
Since capturing temporal dynamics in QoS requirements of sources is allowed in our system model, we incorporate a novel time-coupling constraint in which delay-sensitivity of sources is considered such that a certain end-to-end average delay for each source over a pre-specified time interval is satisfied.
We propose DA-DNUM algorithm, as a dual-based solution, which allocates source rates for the next time interval in a distributed fashion, given the knowledge of network parameters in advance. To overcome the slow convergence of dual-based DA-DNUM algorithm, we propose another fast alternative solution based on the recently-proposed distributed Newton method. Also, we extend and address the problem in a case that the problem data is not known fully in advance to capture more realistic scenarios. 
Through numerical experiments, we show that DA-DNUM gains higher average link utilization and a wider range of feasible scenarios in comparison with the best, to our knowledge, rate control schemes that may guarantee such constraints on delay.
\end{abstract}

}

\maketitle

\IEEEdisplaynotcompsoctitleabstractindextext

%

\IEEEpeerreviewmaketitle

\newtheorem{myTheo}{Theorem}
\newtheorem{myLemma}{Lemma}
\newtheorem{myAssumption}{A}
\newtheorem{myDef}{Definition}

\section{Introduction}
\label{sec:intro}
\IEEEPARstart{N}{owadays}, a plethora of computer applications have emerged that evince delay-sensitivity, mainly in the form of some guarantee on the end-to-end delay. Despite instantaneous delay sensitivity shown by some applications, one may identify several others that only concern the average end-to-end delay over some interval of interest. A notable instance is media streaming where end-to-end delay, averaged over a pre-specified interval, is obliged not to exceed some threshold to ensure continuous playback. Some other examples include real-time WSNs and networked control systems. In such scenarios, due to temporal variations in both source traffic and network characteristics, we face an ever increasing need to accomplish rate allocation capable of capturing such dynamicity.  


As a promising framework, Network Utility Maximization (NUM) has been exploited in several 
network resource allocation scenarios; see, e.g., \cite{chiang2007layering,kelly,Low}. In its simplest form, NUM concerns a network that supports a set of sources and links. Each source is associated with a utility as a function of its rate and transmits its packets through a route, which is a subset of the links in the network. The fixed capacity of links and routing structure dictate a set of linear capacity constraints. The goal of the NUM problem is to find source rates that maximize the aggregate utility of the network given capacity constraints.

A number of studies have thus far incorporated end-to-end delay in the NUM model \cite{Rosberg,li2011congestion,hou2010utility,Dogahe,QiuDelay,xiong2011delay,sadegh,harks}. In these works, end-to-end delay either is included in the objective function of NUM (e.g., \cite{li2011congestion,sadegh,harks}) or introduced some constraints to the NUM problem (e.g., \cite{Rosberg,Dogahe,QiuDelay}). 
In \cite{li2011congestion}, delay is incorporated to the objective function and therefore, delay plays its role as a penalty to the utility function. Based on a delay-sensitive utility function introduced in \cite{shenker}, authors in \cite{sadegh,harks} aim to propose some application-oriented rate allocation schemes employing an alternative utility definition. Both approaches, however, show incompetency to provide some guarantee for delay. 
Despite these studies, NUM framework is intrinsically incapable of capturing temporal variation in network characteristics especially when these characteristics evolve with time scales comparable to those of the underlying dual-based algorithms.  
Generally speaking, (single-period) NUM along with delay constraints is subject to limited degrees of freedom, and as a result, one may face a broad range of infeasible problems.

The conquest of variability-aware NUM-based approaches was further followed by \cite{DNUM}, where it introduced Dynamic NUM (DNUM) as a multi-period extension of NUM. Indeed, DNUM simply considers the network utility aggregated over a finite time interval and thereby takes into account  temporal variations in utilities. Moreover, it allows linear constraints on source rates, referred to as \emph{delivery contracts}, which may be construed as QoS constraints over the time interval. Such delivery contracts, however, are incompetent to capture more complicated key features such as queuing delays and jitter.
In contrast to single-period NUM that suffers from limited degrees of freedom, DNUM offers several flexibilities. In particular, the former may face lots of infeasible problems whereas the latter admits relatively larger set of feasible problems yet higher total aggregated utility.

In this paper, we propose a variant of DNUM that strives to allocate source rates so as to satisfy constraints on end-to-end delays as well as capacity constraints. Toward this, the main contributions of this paper are summarized as follows:

$\vartriangleright$ Built upon DNUM framework, we characterize the average end-to-end delay requirements of sources as a set of general and well-structured constraints. Our model is a generalized version of the model that is built on \cite{QiuDelay,Qiu2013}, and thereby it avoids precise knowledge of underlying packet arrival models and relies only on the first order derivative of the delay function. Generalization of the model of \cite{QiuDelay} to a multi-period setup ushers in several flexibilities. The most promising one, perhaps, is that it allows some degree of freedom to sacrifice utility in some periods so as to maintain delay while compensating for it in some other periods. Secondly, our proposed formulation endows us the ability of maintaining several delay constraints for each source, where each delay constraint concerns a particular time interval of interest. 

$\vartriangleright$ We develop a distributed algorithm  called \emph{Delay-aware Dynamic Network Utility Maximization} (DA-DNUM) that solves the problem granted the knowledge of parameters for the next time interval in advance. Our solution is based on dual decomposition approaches and since we concentrate on strongly convex delay functions and consequently cast the rate allocation as a convex optimization problem, the problem can be efficiently solved in a distributed way thanks to existing dual-based approaches.

$\vartriangleright$ DA-DNUM is in the category of dual decomposition techniques which generally suffer from the curse of slow convergence.  This unpleasant property becomes more salient in the case of solving DNUM problems, where ahead of each time horizon, we must solve the entire problem during all periods. To overcome the slow convergence of the dual-based solution, we devise another solution approach using the recently-proposed distributed Newton method \cite{NUM-Newton} which is a second order algorithm that achieves the optimal solution with faster convergence rate.

$\vartriangleright$ Dependence of DA-DNUM on the precise knowledge of future network parameters stimulates devising another scheme that efficiently work under uncertainty of the parameters. Toward this goal, we also investigate the problem when the problem data is not known fully ahead of time. In particular, we construct another solution based on model predictive control \cite{mpc}, for approximately solving a variation of the problem, in which the link capacities are not known in advance.

$\vartriangleright$ We verify the correctness of our proposed solutions and DA-DNUM algorithm by a set of tractable numerical experiments and give some comparison scenarios to demonstrate its superiority against to the relevant state-of-the-art rate allocation schemes. As an interesting observation, our result corroborates that the proposed temporal formulation enlarges the set of feasible scenarios in comparison with \cite{QiuDelay}. 

The remainder of this paper is organized as follows. First, in Section~\ref{sec:rel} we briefly review the related work. In Section \ref{sec:sysmodel}, we introduce the temporal-aware system model, the characterization of delay constraints, and problem formulation. In Section \ref{sec:sol}, we derive the proposed iterative algorithm and prove the convergence of DA-DNUM. In Section~\ref{sec:newton}, we introduce another solution to the problem to reduce the convergence of the distributed algorithm.  
Section~\ref{sec:mpc} is devoted to introduce a solution when the problem data is not known in advance. 
Section \ref{sec:sim} gives experimental results and Section \ref{sec:conc} is devoted to conclusion and outlining some future directions.

\subsection{Basic Notations and Terminologies}
Throughput the paper, we use the following notations. 
For any vector $\mathbf z$ (matrix $Z$), $\mathbf z\geq 0$ ($Z\geq 0$) means that all components of vector $\mathbf z$ (matrix $Z$) are non-negative. The vector $\mathbf e_j$ denotes the $j$-th unit vector. The operator $\|.\|$ signifies standard Euclidean norm. The domain of a function $f$ is denoted by $\hbox{\textbf{dom }} f$. Moreover, $\mathbbm 1_A$ is 1 if $A$ occurs and 0 otherwise. 
Finally, $[.]^+$ and $[.]_{\mathcal{P}}$ defines the projection onto the positive orthant and set $\mathcal P$, respectively.
We also give some necessary definitions that can be found in, e.g., \cite{Bert_NLP}. 
\begin{myDef}
A function $f(.)$ is a $G$-Lipschitz function if 
$$
|f(\mathbf x_1)-f(\mathbf x_2)|\leq G\|\mathbf x_1-\mathbf x_2\|,\quad \forall \mathbf x_1,\mathbf x_2\in \hbox{\textbf{dom }} f.
$$ 
\end{myDef}

\begin{myDef}
A convex function $f(.)$ is $\kappa$-strongly convex if and only if there exists a constant $\kappa>0$ such that the function $f(\mathbf x)-\frac{\kappa}{2}\|\mathbf x\|^2$ is convex. 
\end{myDef}


It can be easily seen that  if $f(.)$ is convex and satisfies $\|\nabla f(.)\|\leq G$, then it is $G$-Lipschitz. 
We remark that if $f(.)$ is twice differentiable then $f(.)$ is $\kappa$-strongly convex if there exists constant $\kappa$ such that $\nabla^2 f(\mathbf x)-\kappa I$ is positive semidefinite.

\section{Related Work}
\label{sec:rel}
In the recent years, many studies have employed NUM framework to propose efficient protocols and algorithms for network applications under different types of traffics, assumptions, and constraints (see \cite{chiang2007layering} and the references therein). 
In particular, by extending the basic single-period version of NUM framework, a number of studies have incorporated end-to-end delay \cite{Rosberg,li2011congestion,hou2010utility,Dogahe,QiuDelay,xiong2011delay,sadegh,harks} to address the requirements of the delay-sensitive traffic and applications. 
In these works, end-to-end delay either is included in the objective function of NUM (e.g., \cite{li2011congestion,sadegh,harks,pongsajapan2007reverse}) or is augmented as constraints to the underlying optimization problem (e.g., \cite{Rosberg,Dogahe,QiuDelay}). 

\textbf{Delay as objective function.} In \cite{li2011congestion}, delay is incorporated to the objective function and therefore, delay plays its role as a penalty to the utility function. Consequently, the goal is to simultaneously maximize the aggregated utility of all sources and reduce the end-to-end delays. Based on a delay-sensitive utility function introduced in \cite{shenker}, authors in \cite{sadegh,harks} aim to propose some application-oriented rate allocation schemes employing an alternative utility definition. Both approaches, however, show incompetency to provide some guarantee for delay, thereby fail to be employed in QoS-aware applications with hard long term average delay requirements.

\textbf{Delay as constraint.} In another set of works \cite{Rosberg,Dogahe,QiuDelay,jaramillo2011scheduling}, the source delay is incorporated as constraints of the optimization problems. By introducing Virtual Link Capacity Margin (VLCM) to characterize source delay as constraint of the problem, the authors in \cite{QiuDelay,Qiu2013} have proposed a joint rate allocation and scheduling scheme in multi-hop wireless networks. By a different approach in \cite{Dogahe}, another variant of NUM problem is formulated to address joint power and rate control. Moreover, in \cite{Rosberg}, using an elegant fluid model of multi-class flows with different delay requirements, another distributed and stable delay-aware algorithm is proposed.
Despite these single-period NUM-based studies, NUM framework is intrinsically incapable of capturing temporal variations in network characteristics especially when these characteristics evolve with time scales comparable to those of the underlying dual-based algorithms.  
Generally speaking, (single-period) NUM along with delay constraints is subject to limited degrees of freedom, and as a result, one may face a broad range of infeasible problems. We will investigate this phenomenon in details in our experiments in Section~\ref{sec:sim}. 

To capture dynamics in network and sources, NUM framework has been extended to the DNUM framework \cite{DNUM} that supports time-varying characteristics in network model parameters such as flow utilities, links capacities and routing matrix. The DNUM framework has been extended in different research areas \cite{hajiesmaili2012content,wei2012distributed}. In \cite{hajiesmaili2012content}, the time-varying nature is utilized to consider temporal variations in modeling the utility function of the sources with video streaming applications. The authors in \cite{wei2012distributed} have proposed another solution for DNUM based on distributed Newton methods. To the author's knowledge, this is the first work that extends the DNUM framework to characterize the average delay requirements of sources in a general and well-structured way.

\section{Model and Problem Formulation \label{sec:sysmodel}}
\subsection{Network Model}
Our model is based on that of DNUM \cite{DNUM}, which considers rate allocation over a discrete-time interval $\mathcal{T} = \{1,\dots,T\}$\footnote{The duration of each period $t$ and the whole time horizon $T$ is an application-specific design parameter. As an example, in a previous work \cite{hajiesmaili2012content}, video streaming is the underlying application, thus, each period is set according to the length of the video frames and the time horizon $T$ is set according to the length of GOPs (Group Of Pictures). }. We assume that network possesses a set $\mathcal{L} = \{1,\dots,L\}$ of links shared among a set $\mathcal{S} = \{1,\dots,S\}$ of sources. 
We represent the possibly time-varying routing in the network  defined by routing matrices \mbox{$R_t = [(R_t)_{ls}]_{L \times S}, t \in \mathcal{T}$,} whose element $(R_t)_{ls}$ is defined as follows:
  \begin{equation*}
 (R_t)_{ls}=\left \{ \begin{array}{ll}
 1 & \qquad \textrm{if $s$th source passes through $l$ at time $t$}\\
 0 & \qquad \textrm{otherwise}\\	
 \end{array} \right. \nonumber
 \end{equation*}
We let $c_{tl}$ denote the capacity of link $l$ at period $t$ and $\mathbf c_{t} = [c_{tl}]_{l \in \mathcal{L}}$ be the vector of link capacities at period $t$. 

Moreover, we let $x_{st}\in\mathcal X_{st}$ be the transmission rate of source $s$ at period $t$, where we define $\mathcal X_{st} \triangleq [w_{st},W_{st}]$ and $w_{st}$ and $W_{st}$ are the minimum and the maximum rates of source $s$ at period  $t$, respectively. We further require ${0<w_{st}\leq W_{st}, \quad \forall s,t}$. 
We let $X = [x_{st}]_{S \times T}$ be the \emph{rate matrix} and define $\mathcal X=\{X\in \mathbb R^{S\times T}: x_{st} \in \mathcal X_{st}\}$. A feasible rate matrix $X$ then satisfies: $X\in\mathcal X$.  The summary of the main notations of the paper is listed in Table~\ref{tbl:not}. 

\begin{table}[!htp]
\caption{Key Notations}\label{tbl:not}
\centering
\begin{tabular}{|c|L{6.3cm}|}
\hline \textbf{Notation} & \textbf{Definition} \\
\hline $\mathcal{T}$ & The set of time slots (period), $T \triangleq |\mathcal{T}|$\\
$\mathcal{L}$ & The set of links, $L \triangleq |\mathcal{L}|$\\
$\mathcal{S}$ & The set of sources, $S \triangleq |\mathcal{S}|$\\
 \hline\hline
 $R_t$ & The routing matrix at period $t$ \\
 $c_{tl}$ & The capacity of link $l$ at period $t$\\
 $\sigma_{tl}$ & Link margin of link $l$ at period $t$ \\
 $D(\cdot)$ & The delay function of link $l$ at period $t$\\
 \hline\hline
 $x_{st}$ & The transmission rate of source $s$ at  $t$\\
$w_{st}$  & The minimum rate of source $s$ at   $t$ \\
$W_{st}$ & The maximum rate of source $s$ at   $t$ \\
$\phi_{st}$ & The end-to-end queuing delay of source $s$ at period $t$\\
$\mathcal K_s$ & The number of delay constraints of source $s$, $K_s \triangleq \mathcal{K}_s$\\
$M_s$ & The delay indicator matrix of source $s$\\
$d_{sk}$ &  The average delay requirement of source $s$ for its $k$'s delay constraint\\
$U_{st}(\cdot)$ & The utility function of source $s$ at  $t$\\
\hline
\end{tabular}
\end{table}

\subsection{Capacity Constraints}
To give capacity constraints, we first give the definition of \emph{link margin} variables: For each link $l$ and time period $t$, link margin variable $\sigma_{tl}$ is defined as the difference between capacity of link $l$ and the maximum allowable flow passing through it \cite{QiuDelay}. 
Unlike \cite{QiuDelay}, however, our setup does not admit schedulability constraints and hence we proceed to formulate link margin as follows. Consider conventional capacity constraint for link $l$ at period $t$ given by  
$$\sum_{s\in\mathcal S} (R_t)_{ls} x_{st}+\sigma_{tl} = c_{tl}\quad\hbox{ and }\quad \sigma_{tl}\geq 0.$$ 
We then relax the equality constraint above and establish the following constraints for link $l$ at period $t$:
$$\sum_{s\in\mathcal S} (R_t)_{ls} x_{st}+\sigma_{tl}\leq c_{tl}\quad\hbox{ and }\quad \sigma_{tl}\geq 0.$$
The relaxation above, though constricts resource usage (i.e., capacity), plays an important role in limiting the flow of link $l$ and thereby proves essentially useful to control the queuing delay of link $l$. Introducing \mbox{$\boldsymbol \sigma_t = [\sigma_{tl}]_{l \in \mathcal{L}}$} and $\boldsymbol\sigma=[\boldsymbol \sigma_t]_{t\in\mathcal T}$, we then represent the capacity constraints in a compact way as
\begin{eqnarray}
	\label{eq:c1}
	R_t X\mathbf e_t + \boldsymbol \sigma_{t}\leq \mathbf c_{t}\quad \hbox{ and }\quad \boldsymbol \sigma_{t}\geq 0,\quad \forall t \in \mathcal{T}.
\end{eqnarray}
These constraints constitute a set of $2T\times L$ linear inequalities.

\subsection{Average Delay Constraints}
Having defined link margin variables, we define $D(\sigma_{tl})$ as the delay of link $l$ at period $t$. Clearly, the way $D(\sigma_{tl})$ depends on $\sigma_{tl}$ is determined by the packet arrival process model. For instance, for M/M/1 queuing model whose packet arrival is a Poisson process, we have 
\begin{equation}
	\label{eq:mm1}
	D(\sigma_{tl}) = \frac{q}{\sigma_{tl}},\qquad q > 0.
\end{equation}
Another notable instance is the case of M/G/1 queuing model whose delay function is given in \cite{Dogahe,saad2007optimal}. 


In what follows, we list our assumptions on the delay function $D(.)$:
\begin{itemize}
\item[\textbf{A1.}] $D(.)$ is twice differentiable.
\item[\textbf{A2.}] $D(.)$ is $G$-Lipschitz.
\item[\textbf{A3.}] $D(.)$ is $\kappa_D$-strongly convex.  
\end{itemize}


A notable example that satisfies these assumptions is the delay function of (\ref{eq:mm1}). We also remark that these assumptions are valid for M/G/1-based arrival processes, thereby cover the majority of existing queuing models. 

In the present study, we only consider queuing delays and hence, for each source $s$, we obtain the end-to-end delay by simply adding up all link delays along the path of $s$. Writing $\phi_{st}$ for the end-to-end queuing delay of source $s$ at period $t$, we get
$$\phi_{st}= \sum_{l\in\mathcal L} (R_t)_{ls}D(\sigma_{tl}).$$
We further introduce \mbox{$\boldsymbol \phi_{s} = [\phi_{st}]_{t \in \mathcal{T}}$}. Next, we define the constraint on average end-to-end delay as follows: Assume that source $s$ requires its  average end-to-end queuing delay over some interval of interest $\mathcal T_\Delta\subseteq \mathcal T$ with length $\Delta$ be less than some constant $d$. This constraint is formally given by
\begin{equation}
\label{eq:delay_const_prototype}
\frac{1}{\Delta}\sum_{t \in \mathcal{T}_{\Delta}} \phi_{st} \leq d.
\end{equation}

To model a general scenario for the introduced delay constraint, we assume that each source $s$ requires $K_s$ delay constraints of the form (\ref{eq:delay_const_prototype}), indexed by ${k\in\mathcal K_s=\{1,\dots,K_s\}}$. We further introduce a real-world example on realization of this consideration in a typical mission-oriented wireless sensor network scenario in Subsection~\ref{sec:eg}. Each delay constraint $k\in\mathcal K_s$ concerns a specific time interval. Overlap between such intervals, however, is allowed. 
In order to encode delay constraints of the form (\ref{eq:delay_const_prototype}), for each source $s$, we introduce the \emph{delay indicator matrix} \mbox{$M_s = [(M_s)_{kt}]_{K_s \times T}$} as follows
\begin{equation}
(M_s)_{kt}=\left \{ \begin{array}{ll}
\frac{1}{G^s_k} & \qquad \textrm{\small if $k$-th delay constraint of $s$ concerns $t$},\\
0 & \qquad \textrm{\small otherwise,}\\	
\end{array} \right. \nonumber
\end{equation}
where $G^s_k=\sum_{t\in\mathcal T} \mathbbm 1_{\{(M_s)_{kt}\neq 0\}}$. Now, we can write the $k$-th delay constraint of source $s$ as 
$$\sum_{t\in\mathcal T} (M_s)_{kt} \phi_{st}\leq  d_{sk},$$
where $d_{sk}$ is the average delay requirement of source $s$ for its $k$'s delay constraint. Note that the elements of every row of $M_s$ add up to 1 and therefore, we may interpret the left hand side of the constraint above, like that of (\ref{eq:delay_const_prototype}), as the end-to-end queueing delay of $s$ averaged over time interval $\{t\in\mathcal T: (M_s)_{kt}=1\}$.  
Moreover, letting $\mathbf d_s = [d_{sk}]_{k \in\mathcal K_s}$ yields the following vector representation for delay constraints:
\begin{eqnarray}
	\label{eq:c2}
	M_{s} \boldsymbol \phi_{s} \leq \mathbf d_s, \qquad \forall s \in \mathcal{S}.
\end{eqnarray}

These constraints constitute a set of $\sum_{s\in\mathcal S} K_s$ inequalities that are nonlinear in $\boldsymbol\sigma$.

\subsubsection{An Illustrative Example: Mission-Oriented WSNs}
\label{sec:eg}
To motivate the appropriateness of the model above, we next provide a practical application of this model for mission-oriented wireless sensor networks (WSN) \cite{eswaran2012utility}, i.e. the case where there are several coexisting applications (henceforth \emph{missions}) in a WSN. Let us look at a surveillance application that employs various types of sensors such as \emph{video}, \emph{motion detector}, and \emph{thermal sensors} to provide assistive ambient intelligence in e.g., disaster recovery environments. 

The naive approach is to require each sensor  to periodically transmit the data at specific time intervals.
Albeit simple to implement, this approach is inefficient as each mission might possess particular QoS requirement in terms of end-to-end delay. For instance, a video mission may demand for a long-time delay constraint to work efficiently. In contrast, the thermal mission may report the temperature periodically on a regular basis and thereby declares a short-term delay requirement at certain periods.


The network designer therefore needs to select network parameters properly to achieve the best efficiency. 
Besides other parameters, one could set $\mathcal{T}_{\Delta} = \mathcal{T}$ for the real-time video mission, as it records and streams data to the sink continuously. The value of $\mathcal{T}_{\Delta}$ has a periodic shape for the thermal sensor. Say, in the case of $T = 60$, we can define $\mathcal{T}_{\Delta_1} = \{1,2,3\}$, $\mathcal{T}_{\Delta_2} = \{21,22,23\}$, and $\mathcal{T}_{\Delta_3} = \{41,42,43\}$. In this respect, this sensor reports its data in 3 different steps as mentioned above. 

In summary, one can identify several other application scenarios (such as in emerging Internet of Things (IoT) or networked control systems), wherein different competing goals (missions in some contexts) with diverse QoS characteristics coexist under a unified application, but, with heterogeneous requirements.

\subsection{Optimization Problem}
We associate a utility function $U_{st}(x_{st})$ to each source $s$ at period $t$.
Assumptions on the utility functions are:
\begin{itemize}
\item[\textbf{A4.}] For every $s$ and $t$, $U_{st}(.)$ is continuous, monotonically increasing, and twice differentiable. 
\item[\textbf{A5.}] For every $s$ and $t$, $-U_{st}(.)$ is $\kappa_U$-strongly convex.
\end{itemize}

Similar to \cite{DNUM}, we define the network utility $U(.)$ as the sum of all utilities over time horizon $\mathcal T$ as follows:
$$U(X) = \sum_{s\in\mathcal S} \sum_{t\in\mathcal T} U_{st}(x_{st}).$$

Now, we cast the rate allocation problem as 
\begin{align*}
 \textsf{P1:}&  \quad  \max_{X\in\mathcal{X},\boldsymbol \sigma\ge 0}  \quad U(X)\\
	&\textrm{subject to:} \\
	&\qquad\quad R_t X\mathbf e_t +\boldsymbol \sigma_{t}\leq \mathbf c_{t}, \qquad\forall  t \in \mathcal{T},\\ 
&\qquad\quad M_{s} \boldsymbol \phi_{s} \leq \mathbf d_s, \qquad\qquad\forall s \in \mathcal{S},\\
&\qquad\quad\phi_{st} = \sum_{l\in\mathcal L} (R_t)_{ls}D(\sigma_{tl}), \quad\forall s \in \mathcal{S}, \forall t\in\mathcal T. 
\end{align*}

First we highlight that constraints of \textsf{P1} constitute a compact set. Hence, at least one optimal solution exists. Furthermore, \textsf{P1} is a strongly convex optimization problem. An immediate consequence of this property is that the optimal solution is unique. 
We remark that \textsf{P1} is non-separable due to coupled delay constraints. It's worth noting that in the lack of average delay constraints, problem \textsf{P1} degenerates to DNUM problem of \cite{DNUM} without delivery contracts. In the above formulation, we address QoS requirements mainly through end-to-end delay constraints and thus avoid augmenting delivery contracts, i.e. linear  constraints on source rates over $\mathcal T$. We stress, however, that the solution procedure below permits having delivery contracts as well. 
We further note that for the case of $T=1$ and $K_s=1,\forall s$, \textsf{P1} reduces to problem formulation in \cite{QiuDelay} (for the case of rate allocation only). 

\section{Optimal Rate Allocation Algorithm \label{sec:sol}}
\begin{figure*}
\begin{align}
\label{eq:p1_lag}
L(X, \boldsymbol \sigma, \boldsymbol \lambda, \boldsymbol \mu) &= \sum_{t\in\mathcal T} \sum_{s\in\mathcal S} U_{st}(x_{st}) - \sum_{t\in\mathcal T}  \boldsymbol \lambda_{t}^{\textsf{T}} \left( R_t X\mathbf e_t - \mathbf c_{t} + \boldsymbol \sigma_{t} \right)- \sum_{s\in\mathcal S}  \boldsymbol \mu_{s}^{\textsf{T}} \left(M_{s} \boldsymbol \phi_{s} - \mathbf d_s\right)\\
 &= \sum_{t\in\mathcal T} \sum_{s\in\mathcal S} \left(U_{st}(x_{st}) - \lambda^{st} x_{st}\right)- \sum_{t\in\mathcal T} \sum_{l\in\mathcal L} \left(\mu^{tl} D(\sigma_{tl}) + \lambda_{tl} \sigma_{tl} \right)+ \sum_{t\in\mathcal T}  \boldsymbol\lambda_{t}^{\textsf{T}}\mathbf c_{t} + \sum_{s\in\mathcal S}  \boldsymbol \mu_{s}^{\textsf{T}} \mathbf d_s. \nonumber
\end{align}
\begin{align}
\label{eq:dual_func}
D(\boldsymbol \lambda, \boldsymbol \mu) = \max_{X\in\mathcal{X}, \boldsymbol \sigma\geq 0} L(X, \boldsymbol \sigma, \boldsymbol \lambda, \boldsymbol \mu) =\max_{X\in\mathcal{X}} \sum_{t\in\mathcal T} \sum_{s\in\mathcal S} \left(U_{st}(x_{st}) - \lambda^{st} x_{st}\right) + \max_{\boldsymbol\sigma\geq 0}  \sum_{t\in\mathcal T} \sum_{l\in\mathcal L} \left(\mu^{tl} D(\sigma_{tl}) + \lambda_{tl} \sigma_{tl} \right).
\end{align}
\hrule
\hrulefill
\end{figure*}
In this section, we solve \textsf{P1} and develop a distributed rate allocation algorithm. We note that strong duality \cite{Boyd} holds for \textsf{P1} and hence we can solve it through its dual. We let \mbox{$\boldsymbol \lambda_{t} = [\lambda_{tl}]_{l \in \mathcal{L}}$} 
and \mbox{$\boldsymbol \mu_{s} = [\mu_{sk}]_{k\in\mathcal K_s}$} respectively denote the Lagrange multipliers (dual variables) associated to the capacity constraints for period $t$ and average delay constraints for source $s$. Moreover, we introduce \mbox{$\boldsymbol \lambda = [\boldsymbol \lambda_{t}]_{t \in \mathcal{T}}$} and \mbox{$\boldsymbol \mu = [\boldsymbol \mu_{s}]_{s \in \mathcal{S}}$}. Now, we give the partial Lagrangian of \textsf{P1} in (\ref{eq:p1_lag}), where 
\begin{align*}
	&\lambda^{st} \triangleq \sum_{l\in\mathcal L} (R_t)_{ls} \lambda_{tl},\\
	&\mu^{tl} \triangleq \sum_{s\in\mathcal S} \sum_{k\in\mathcal K_s} (M_s)_{kt} (R_t)_{ls} \mu_{sk}\nonumber.
\end{align*}

To solve problem \textsf{P1}, we derive the dual function $g(\boldsymbol \lambda, \boldsymbol \mu)$ in (\ref{eq:dual_func}) and establish the dual problem associated to \textsf{P1} as \cite{Bert_NLP}:
\begin{eqnarray}
	\label{eq:d1}
	\textsf{D1}:\quad \min_{\boldsymbol \lambda\geq 0,\boldsymbol\mu \geq 0} g(\boldsymbol \lambda, \boldsymbol \mu). \nonumber
\end{eqnarray}

Given $\boldsymbol \lambda$ and $\boldsymbol \mu$, let \mbox{$X^{\star} = [x_{st}^{\star}]_{T\times S}$} and \mbox{$\boldsymbol \sigma_t^{\star} = [\sigma_{tl}^{\star}]_{l \in \mathcal{L}}$} be the maximizers of maximization problems in (\ref{eq:dual_func}). To derive these solutions, first note that partial derivatives of the Lagrangian are given by: 
\begin{align*}
	\frac{\partial L}{\partial x_{st}} = U'_{st}(x_{st}) - \lambda^{st}, \quad \forall s,\forall t\\
	\frac{\partial L}{\partial \sigma_{tl}} = \mu^{tl}D'(\sigma_{tl}) + \lambda_{tl},\quad \forall t,\forall l.
\end{align*}
The maximizers are stationary point of the Lagrangian. Therefore, through preliminary manipulations we get  
\begin{align*}
	x_{st}^{\star}(\boldsymbol\lambda)=\left[U'^{-1}_{st}(\lambda^{st})\right]_{\mathcal{X}_{st}},\quad \forall s,\forall t\\
\sigma_{tl}^{\star}(\boldsymbol\lambda,\boldsymbol\mu)=\left[D'^{-1}\left(-\frac{\lambda_{tl}}{\mu^{tl}}\right)\right]^+, \quad \forall t,\forall l.
\end{align*}
One consequence of strong convexity of \textsf{P1} is that the dual function $g(\boldsymbol \lambda, \boldsymbol \mu)$ is differentiable in its domain. Hence, we can employ the \emph{gradient projection method} \cite{Bert_NLP} to solve \textsf{D1}. 
Using Danskin's Theorem \cite{Bert_NLP}, partial derivatives of dual function $g(\boldsymbol \lambda, \boldsymbol \mu)$ are given by:
\begin{equation}
	\label{eq:d1}
\frac{\partial g}{\partial \lambda_{tl}}=c_{tl} - \sigma_{tl}-\sum_{t\in\mathcal T} \sum_{s\in\mathcal S} (R_t)_{ls} x_{st},\quad \forall l\in\mathcal L, \forall t\in\mathcal T,
\end{equation}
\begin{equation}
	\label{eq:d2}
\frac{\partial g}{\partial \mu_{sk}}= d_{sk}-\sum_{t\in\mathcal T} \sum_{l\in\mathcal L} (R_t)_{ls} D(\sigma_{tl}),\quad \forall s\in\mathcal S, \forall k\in\mathcal K_s.
\end{equation}
Using these, for dual variable update needed for gradient projection method we get
\begin{align*}
\lambda_{tl}^{(j+1)} = \bigg[\lambda_{tl}^{(j)} + \gamma &\left(  \sum_{s\in\mathcal S} (R_t)_{ls} x_{st}^{(j)}  + \sigma_{tl}^{(j+1)}- c_{tl} \right) \bigg]^+,\\
&\qquad\qquad\qquad\qquad\qquad\forall l\in\mathcal L, \forall t\in\mathcal T,\\ 
\mu_{sk}^{(j+1)}= \bigg[\mu_{sk}^{(j)} + \gamma &\left( \sum_{t\in\mathcal T} \sum_{l\in\mathcal L} (R_t)_{ls} D(\sigma_{tl}^{(j)}) - d_{sk} \right) \bigg]^+,\\
&\qquad\qquad\qquad\qquad\qquad\forall s\in\mathcal S, \forall k\in\mathcal K_s,
\end{align*}
%
where $x^{(j)}_{st}=x_{st}^\star(\boldsymbol\lambda^{(j)})$, $\sigma^{(j+1)}_{tl}=\sigma_{tl}^\star(\boldsymbol\lambda^{(j)},\boldsymbol\mu^{(j)})$, and $\gamma >0$ is a sufficiently small step size. 

Note that proper selection of step size $\gamma$ is crucial for guaranteeing the convergence of the iterative solution above. 

\begin{myTheo}
Assume that $\gamma$ satisfies $0<\gamma <\frac{2}{Q}$, where 
\begin{eqnarray}
Q&\triangleq& TL\left(\frac{1}{\mu_{\min}\kappa_D}+\frac{S}{\kappa_U}\right)+ GTL\frac{\lambda_{\max}}{\mu_{\min}^2 \kappa_D}\sum_s{K_s} 
\nonumber \\ 
&+& \sqrt{TL\sum_s K_s}\left(\frac{G}{\mu_{\min}\kappa_D}+\frac{1}{\mu_{\min}^2 \kappa_U}\right),
\end{eqnarray}
$$\lambda_{\max}=\max_{t,l} \lambda_{tl},\quad \mu_{\min}=\min_{s,k\in\mathcal K_s} \mu_{sk}.$$ 
Then, starting from any initial point, the limit point $(X^\star,\sigma^\star,\mu^\star,\lambda^\star)$ of the sequence $\{X^{(j)},\sigma^{(j)},\mu^{(j)},\lambda^{(j)}\}_{j\geq 1}$ generated by the aforementioned iterative solution is primal-dual optimal and $(X^\star,\sigma^\star)$ is the unique optimal solution to \textsf{P1}.
\end{myTheo}


\begin{proof}
First, we briefly review the descent lemma for solving the problem $\min_{\mathbf z} f(\mathbf z)$ using gradient method with constant step size \cite{Bert_NLP}. We let $\mathbf z^\star$ be the minimizer of the problem. If $\nabla f(\mathbf z)$ is a  $Q$-Lipschitz function, 
then the sequence $\{\mathbf z^{(k)}\}_{k\geq 0}$ defined by   
\begin{equation}
\mathbf z^{(k+1)}=\mathbf z^{(k)}-\gamma\nabla f(\mathbf z^{(k)})
\end{equation}
converges to $\mathbf z^\star$ provided that $0< \gamma <\frac{2}{Q}.$ By this lemma, to prove the convergence of the algorithm, it suffices to find constant $Q$ that satisfies Lipschitz condition. Let us define $\nu=[\boldsymbol\lambda_1,\dots,\boldsymbol\lambda_T,\boldsymbol\mu_1,\dots,\boldsymbol\mu_S]$. Then, we should find constant $Q$ such that $\nabla g(\boldsymbol\nu)$ is $Q$-Lipschitz. Equivalently, we can resort to find an upper bound for the $\ell_2$-norm of the Hessian of $g(\boldsymbol\nu)$. 
The Hessian of $g(\boldsymbol\nu)$, henceforth denoted by $H$, is a $(TL+\sum_s K_s)$-by-$(TL+\sum_s K_s)$ matrix, whose $ij$-element is:
$$
H_{ij}=\frac{\partial^2{g(\boldsymbol\nu)}}{\partial\nu_i\partial\nu_j}.
$$
 
Recall that the partial derivatives of $g(\boldsymbol\nu)$ w.r.t. dual variables are given by equations (\ref{eq:d1}) and (\ref{eq:d2}).
Before proceeding to calculate the elements of the Hessian, observe that

\begin{align*}
\frac{\partial x_{st}}{\partial \lambda^{st}}&=\frac{\partial U'^{-1}_{st}(\lambda^{st})}{\partial \lambda^{st}} =\frac{1}{U''_{st}(x_{st})}, \quad \forall s\in\mathcal S,\forall t\in\mathcal T,\\
\frac{\partial \sigma_{tl}}{\partial \lambda_{tl}}&=\frac{\partial }{\partial \lambda_{tl}} D'^{-1}\left(-\frac{\lambda_{tl}}{\mu^{tl}}\right)=-\frac{1}{\mu^{tl}D''(\sigma_{tl})},\quad\forall t\in\mathcal T,\forall l\in\mathcal L,\\
\frac{\partial \sigma_{tl}}{\partial \mu^{tl}}&=\frac{\partial }{\partial \mu^{tl}} D'^{-1}\left(-\frac{\lambda_{tl}}{\mu^{tl}}\right)=\frac{\lambda_{tl}}{(\mu^{tl})^2 D''(\sigma_{tl})},\quad\forall t\in\mathcal T,\forall l\in\mathcal L. 
\end{align*}

Using these equations, for the elements of the Hessian we obtain
	\begin{eqnarray}
		\frac{\partial^2 g(\boldsymbol\nu)}{\partial \lambda_{t'l'} \partial \mu_{sk}} &=& -\frac{\partial}{\partial \lambda_{t'l'}} \sum_{t \in \mathcal{T}}\sum_{l \in \mathcal{L}} (R_t)_{ls} D(\sigma_{tl})  \nonumber \\
		&=& -(R_{t'})_{l's} \cdot \frac{\partial \sigma_{t'l'}}{\partial \lambda_{t'l'}} \cdot \frac{dD(\sigma_{t'l'})}{d\sigma_{t'l'}}  \nonumber \\
\label{eq:H_lam_mu}
		&=& \frac{(R_{t'})_{l's}}{\mu^{t'l'}}\cdot\frac{D'(\sigma_{t'l'})}{D''(\sigma_{t'l'})}, \nonumber
\end{eqnarray}

\begin{eqnarray}
		\lefteqn{\frac{\partial^2 g(\boldsymbol\nu)}{\partial \mu_{s'k'} \partial \mu_{sk}} = -\frac{\partial}{\partial \mu_{s'k'}} {\sum_{t \in \mathcal{T}}\sum_{l \in \mathcal{L}} (R_t)_{ls} D(\sigma_{tl})} }\nonumber \\
		&=& -\sum_{t \in \mathcal{T}} \sum_{l \in \mathcal{L}} (R_t)_{ls} \frac{\partial\sigma_{tl}}{\partial \mu_{s'k'}} \cdot\frac{dD(\sigma_{tl})}{d\sigma_{tl}}  \nonumber \\
		&=& -\sum_{t \in \mathcal{T}} \sum_{l \in \mathcal{L}} (R_t)_{ls} \frac{\partial \mu^{tl}}{\partial \mu_{s'k'}} \cdot\frac{\partial\sigma_{tl}}{\partial \mu^{tl}}\cdot\frac{dD(\sigma_{tl})}{d\sigma_{tl}} \nonumber \\
\label{eq:H_mu_mu}
		&=& -\sum_{t \in \mathcal{T}} \sum_{l \in \mathcal{L}} (R_t)_{ls} (M_{s'})_{k't} (R_t)_{l's}\cdot \frac{\lambda_{tl}}{(\mu^{tl})^2}\cdot\frac{D'(\sigma_{tl})}{D''(\sigma_{tl})}, \nonumber
	\end{eqnarray}
	
\begin{eqnarray}
		\frac{\partial^2 g(\boldsymbol\nu)}{\partial \lambda_{t'l'} \partial \lambda_{tl}} &=& -\frac{\partial \sigma_{tl}}{\partial \lambda_{t'l'}} - \sum_{s \in \mathcal{S}} (R_t)_{ls} \frac{\partial x_{st}}{\partial \lambda_{t'l'}} \nonumber \\
		&=&\frac{\mathbbmss{1}_{\{(t', l') = (t,l)\}}}{\mu^{tl}D''(\sigma_{tl})} - \sum_{s \in \mathcal{S}} (R_t)_{ls} \frac{\partial \lambda^{st}}{\partial \lambda_{t'l'}} \cdot\frac{\partial x_{st}}{\partial \lambda^{st}} \nonumber \\
\label{eq:H_lam_lam}
		&=& \frac{\mathbbm{1}_{\{(t', l') = (t,l)\}}}{\mu^{tl} D''(\sigma_{tl})} - \sum_{s \in \mathcal{S}}  \frac{(R_t)_{ls} (R_{t'})_{l's}}{{U''_{st}}(x_{st})},\nonumber
\end{eqnarray}

	\begin{eqnarray}
	\frac{\partial^2 g(\boldsymbol\nu)}{\partial \mu_{s'k'} \partial \lambda_{tl}} &=& -\frac{\partial \sigma_{tl}}{\partial \mu_{s'k'}} - \sum_{s \in \mathcal{S}} (R_t)_{ls} \frac{\partial x_{st}}{\partial \mu_{s'k'}} \nonumber \\
	&=& - \frac{\partial \mu^{tl}}{\partial \mu_{s'k'}} \cdot \frac{\partial\sigma_{tl}}{\partial \mu^{tl}} \nonumber \\
\label{eq:H_mu_lam}
		&=& -\frac{(M_{s'})_{k't} (R_t)_{ls'} \lambda_{tl}}{(\mu_{tl})^2} \frac{1}{D''(\sigma_{tl})}. \nonumber
	\end{eqnarray}

Using introduced values in the theorem, we can easily establish the following bounds:

\begin{eqnarray}
\label{eq:H_lam_mu_bound}
\bigg|\frac{\partial^2 g(\boldsymbol\nu)}{\partial \lambda_{t'l'} \partial \mu_{sk}}\bigg| &\leq& \frac{G}{\mu_{\min}\kappa_D},  \\
\label{eq:H_mu_mu_bound}
\bigg|\frac{\partial^2 g(\boldsymbol\nu)}{\partial \mu_{s'k'} \partial \mu_{sk}}\bigg| &\leq& \frac{\lambda_{\max}G}{\mu_{\min}^2 \kappa_D} \sum_{t\in\mathcal T}\sum_{l\in\mathcal L} (M_{s'})_{k't}(R_t)_{ls}(R_t)_{ls} \nonumber\\ &\leq& \frac{\lambda_{\max}GTL}{\mu_{\min}^2 \kappa_D},\\
\label{eq:H_lam_lam_bound}
\bigg|\frac{\partial^2 g(\boldsymbol\nu)}{\partial \lambda_{t'l'} \partial \lambda_{tl}}\bigg| &\leq& \frac{1}{\mu_{\min} \kappa_D}+\frac{1}{\kappa_U}\sum_{s\in\mathcal S} (R_t)_{ls}(R_{t'})_{l's} \nonumber \\ &\leq& \frac{1}{\mu_{\min}\kappa_D}+\frac{S}{\kappa_U}\\
\label{eq:H_mu_lam_bound}
\bigg|\frac{\partial^2 g(\boldsymbol\nu)}{\partial \mu_{s'k'}\partial \lambda_{tl} }\bigg| &\leq&  \frac{1}{\mu_{\min}^2 \kappa_D}.
\end{eqnarray}

Next, we find an upper bound for the $\ell_2$-norm of the Hessian $H$. 
To simplify the analysis, we use the following result to upper bound $\|H\|_2$. Recall that for matrices $A_i\in \mathbb R^{n\times n},i=1,\dots,K$, we have\cite{horn2012matrix}:
$$\|\sum_{i=1}^K A_i\|\leq \sum_{i=1}^K \|A_i\|.$$
 
In light of this result, we consider the following decomposition of $H$:
\begin{eqnarray}
H &=&\left[
\begin{array}{c|c}
H^{(\lambda\lambda)} & H^{(\mu\lambda)} \\ \hline
H^{(\lambda\mu)} & H^{(\mu\mu)}
\end{array}\right] \nonumber \\&=&\left[
\begin{array}{c|c}
H^{(\lambda\lambda)} & 0 \\ \hline
0 & 0
\end{array}\right]+\left[
\begin{array}{c|c}
0 & H^{(\mu\lambda)} \\ \hline
0 & 0
\end{array}\right] \nonumber \\ &+& \left[
\begin{array}{c|c}
0 & 0 \\ \hline
H^{(\lambda\mu)} & 0
\end{array}\right]  +\left[
\begin{array}{c|c}
0 & 0 \\ \hline
0 & H^{(\mu\mu)}
\end{array}\right],
\end{eqnarray}

which yields:
\begin{align}
\label{eq:H2_1}
\|H\|_2 &\leq& \|\left[
\begin{array}{c|c}
H^{(\lambda\lambda)} & 0 \\ \hline
0 & 0
\end{array}\right]\|_2+\|\left[
\begin{array}{c|c}
0 & H^{(\mu\lambda)} \\ \hline
0 & 0
\end{array}\right]\|_2 \nonumber \\ &+& \|\left[
\begin{array}{c|c}
0 & 0 \\ \hline
H^{(\lambda\mu)} & 0
\end{array}\right]\|_2+\|\left[
\begin{array}{c|c}
0 & 0 \\ \hline
0 & H^{(\mu\mu)}
\end{array}\right]\|_2.
\end{align}
Moreover, recall that for any matrix $A$ we have: $\|A\|_2\leq \sqrt{\|A\|_1\|A\|_\infty}$, where $\|A\|_1$ is the maximum column-sum matrix norm of  $A$, and $\|A\|_\infty$ is the maximum row-sum matrix norm \cite{horn2012matrix}. Then, since 
for $p=1,\infty$, we have that
$$\|\left[
\begin{array}{c|c}
H^{(\lambda\lambda)} & 0 \\ \hline
0 & 0
\end{array}\right]\|_p=\|H^{(\lambda\lambda)}\|_p,$$
we can further provide an upper bound for the right hand side of (\ref{eq:H2_1}) as follows: 
\begin{align}
\label{eq:H2_2}
\|H\|_2 &\leq& \sqrt{\|H^{(\lambda\lambda)}\|_1\|H^{(\lambda\lambda)}\|_\infty}+\sqrt{\|H^{(\mu\lambda)}\|_1\|H^{(\mu\lambda)}\|_\infty}
\nonumber \\ &+& \sqrt{\|H^{(\lambda\mu)}\|_1\|H^{(\lambda\mu)}\|_\infty}+\sqrt{\|H^{(\mu\mu)}\|_1\|H^{(\mu\mu)}\|_\infty}
\end{align}

We next obtain upper bounds to each term in the right hand side of (\ref{eq:H2_2}). For $H^{(\lambda\lambda)}$, using (\ref{eq:H_lam_lam}) and (\ref{eq:H_lam_lam_bound}), we get
$$
\|H^{(\lambda\lambda)}\|_1\leq TL\left(\frac{1}{\mu_{\min}\kappa_D}+\frac{S}{\kappa_U}\right),$$ 
$$\|H^{(\lambda\lambda)}\|_\infty\leq TL\left(\frac{1}{\mu_{\min}\kappa_D}+\frac{S}{\kappa_U}\right).$$
From (\ref{eq:H_mu_mu}) and (\ref{eq:H_mu_mu_bound}), for $H^{(\mu\mu)}$s we get
$$
\|H^{(\mu\mu)}\|_1\leq GTL\frac{\lambda_{\max}}{\mu_{\min}^2 \kappa_D}\sum_s{K_s},$$
$$ \|H^{(\mu\mu)}\|_\infty\leq GTL\frac{\lambda_{\max}}{\mu_{\min}^2 \kappa_D}\sum_s{K_s}.
$$
For $H^{(\mu\lambda)}$, (\ref{eq:H_lam_mu}) and (\ref{eq:H_lam_mu_bound}) yield
$$
\|H^{(\mu\lambda)}\|_1\leq \frac{G}{\mu_{\min} \kappa_D} \sum_s K_s, $$ 
$$\|H^{(\mu\lambda)}\|_\infty\leq \frac{GTL}{\mu_{\min} \kappa_D}.
$$
And using (\ref{eq:H_mu_lam}) and (\ref{eq:H_mu_lam_bound}) for $H^{(\lambda\mu)}$, we obtain
$$
\|H^{(\lambda\mu)}\|_1\leq  \frac{TL}{\mu_{\min}^2 \kappa_D}, $$
$$ \|H^{(\lambda\mu)}\|_\infty\leq \frac{1}{\mu_{\min}^2 \kappa_D}\sum_s K_s.
$$
Let 
\begin{align*}
Q&\triangleq TL&\left(\frac{1}{\mu_{\min}\kappa_D}+\frac{S}{\kappa_U}\right)+ GTL\frac{\lambda_{\max}}{\mu_{\min}^2 \kappa_D}\sum_s{K_s} 
\nonumber \\ &+& \sqrt{TL\sum_s K_s}\left(\frac{G}{\mu_{\min}\kappa_D}+\frac{1}{\mu_{\min}^2 \kappa_U}\right).
\end{align*}
Then, using (\ref{eq:H2_2}) we obtain
\begin{align}
\|H\|_2 \leq Q.
\end{align}
Hence, the right hand side of the above equation can be viewed as the Lipschitz constant for $\nabla g(\boldsymbol\nu)$, and consequently, if we require 
$0<\gamma <\frac{2}{Q},$
then the algorithm converges to the primal-dual point of \textsf{P1}.

\end{proof}

Given appropriate $\gamma$, update equations for dual variables converge to minimizers of \textsf{D1}. Strong duality then guarantees that optimal values of \textsf{D1} and \textsf{P1} coincide and that $X^\star$ and $\boldsymbol\sigma^\star$ can be obtained accordingly.  
Next, we give a distributed iterative algorithm, named \emph{Delay-Aware Dynamic Network Utility Maximization (DA-DNUM)}, that is based on a distributed implementation of the above iterative solution. 
Since gradient-based algorithms are not finitely convergent, in DA-DNUM algorithm we introduce a parameter \texttt{th} to stop the iterative procedure. 
DA-DNUM algorithm relies on both the knowledge of network parameters in advance of time interval $\mathcal T$ and ability of explicit/implicit exchange of dual variables between sources and links (more precisely,  between each source $s$ and links on the path of $s$). The pseudo-code of DA-DNUM is shown as Algorithm 1. 

\begin{algorithm}[!h]\small
\caption{DA-DNUM Algorithm} \DontPrintSemicolon 
Acquire network parameters for the next time horizon $\mathcal T$. \\
Initialize  $X^0,\boldsymbol\sigma^0,\boldsymbol\lambda^0,$ and $\boldsymbol \mu^0$.
\BlankLine
 \While{$\max\limits_{s,l,t} \left\{|x_{st}^{(j+1)} - x_{st}^{(j)}|, |\sigma_{tl}^{(j+1)} - \sigma_{tl}^{(j)}|\right\} \leq \texttt{th}$}{
{At each link $l$, for each period $t$, obtain ${\mu^{tl,(j)}}$ and update:}\\
\BlankLine\vspace{.2mm}
  \quad $\sigma_{tl}^{(j+1)} = \left[D'^{-1}\left(-\frac{\lambda_{tl}^{(j)}}{\mu^{tl,(j)}}\right)\right]^+$\\
\BlankLine\vspace{.2mm}
 \quad $\lambda_{tl}^{(j+1)} = \bigg[\lambda_{tl}^{(j)} + \gamma \left(\sum_{s\in\mathcal S} (R_t)_{ls} x_{st}^{(j)} + \sigma_{tl}^{(j+1)}- c_{tl}\right) \bigg]^+$ \\
\BlankLine\vspace{.2mm}
 {At each source $s$, for each period $t$, obtain $\lambda^{st,(j)}$ and compute:}\\
\BlankLine\vspace{.2mm}
 \quad $x_{st}^{(j+1)} = \left[U'^{-1}_{st}\left(\lambda^{st, (j)}\right)\right]_{\mathcal{X}_{st}}$ \\
 \quad $\mu_{sk}^{(j+1)} = \left[\mu_{sk}^{(j)} + \gamma \left(\sum_{t\in\mathcal T} \sum_{l\in\mathcal L} (R_t)_{ls} D(\sigma_{tl}^{(j)}) - d_{sk} \right) \right]^+$\\
}
\label{bottomUpProcedure}
\end{algorithm}

DA-DNUM is devised by solving rate allocation problem using first order methods. In general, first order methods suffer from slow rate of convergence. This unpleasant property becomes more salient in the case of solving DNUM problems, where ahead of each time horizon $T$, we must solve the entire problem during all $t$s. There is a fast alternative to solve problem \textsf{P1}; the second order algorithms \cite{NUM-Newton} achieve the optimal point with faster convergence rate. In a nutshell, in state-of-the-art distributed Newton method the direction of dual  adjustment is the negative gradient scaled by the inverse of the Hessian of $\nabla^2 D$ which results in a faster convergence compared to the first order gradient based algorithms.

\section{Distributed Newton Method: A Fast Solution \label{sec:newton}}
In this section, we develop an alternative solution to problem \textsf{P1} that converges substantially faster at the expense of doing more computations. Our work toward this goal is to employ the recently-proposed distributed Newton method in \cite{NUM-Newton} and extend it to consider our problem. 

\subsection{Reformulation of problem \textsf{P1}}
In order to solve the problem \textsf{P1} using distributed Newton method, we first reformulate problem \textsf{P1} so as to possess only equality constraints. 

For the sake of proper explanation of problem \textsf{P1} by matrix notations, we represent the routing matrix \mbox{$R_{TL\times TS}$} as  Eq.~\eqref{eq:RN}.
\begin{figure*}
\begin{eqnarray}
\label{eq:RN}
R = \left[ 
\begin{matrix}
 	(R_1)_1 & 0_{L \times (T-1)} & (R_1)_2 & 0_{L \times (T-1)}  & \dots & (R_1)_S & 0_{L \times (T-1)} \\
	0_{L \times 1} & (R_2)_1 & 0_{L \times (T-1)} & (R_2)_2 & \dots & (R_2)_S & 0_{L \times 1} \\
	\vdots & \vdots & \vdots & \vdots & \ddots & \vdots & \vdots \\
	0_{L \times (T-1)} & (R_T)_1 & 0_{L \times (T-1)} & (R_T)_2 & \dots & 0_{L \times (T-1)} & (R_T)_S
\end{matrix}
\right]
\end{eqnarray}
\hrule
\end{figure*}


Similarly, we define a block diagonal matrix $M_{K \times TS}$ for the second set of end-to-end average constraints as follows:

\begin{eqnarray}
M = \left[ 
\begin{matrix}
 	M_1 & 0_{K_2 \times T} & \dots & 0_{K_S \times T} \\
	0_{K_1 \times S} & M_2 & \dots & 0_{K_S \times T} \\
	\vdots & \vdots & \ddots & \vdots \\
	0_{K_1 \times S} & 0_{K_2 \times T} & \dots & M_S
\end{matrix}
\right],
\end{eqnarray}
with $K= \sum_{s=1}^S K_s$. Moreover, we define the following vectors: 

$\vartriangleright$ Source rate vector of length $TS$ at all time periods denoted by \mbox{$\boldsymbol x = [\boldsymbol{x}_t]_{t \in \mathcal{T}}$}.

$\vartriangleright$ Delay upper bound vector of length $K$ for all sources denoted by  \mbox{$\boldsymbol d = [\boldsymbol{d}_s]_{s \in \mathcal{S}}$}.

$\vartriangleright$ Delay vector of length $TS$ for all sources denoted by  \mbox{$\boldsymbol \phi = [\boldsymbol{\phi}_s]_{s \in \mathcal{S}}$}.

$\vartriangleright$ Link capacity vector of length $TL$ at all time periods denoted by \mbox{$\boldsymbol c = [\boldsymbol{c}_t]_{t \in \mathcal{T}}$}.

$\vartriangleright$ Link margin vector of length $TL$ at all time periods denoted by \mbox{$\boldsymbol \sigma = [\boldsymbol{\sigma}_t]_{t \in \mathcal{T}}$}.

We next rewrite problem \textsf{P1} using the introduced vectors as 

\begin{align}
 \textsf{P2:}&  \quad  \max_{X\in\mathcal{X},\boldsymbol \sigma\ge 0}  \quad U(X)\nonumber\\
 \label{eq:p2_c1}
	&\textrm{subject to:}\nonumber\\
	 &\qquad\quad R \boldsymbol x + \boldsymbol \sigma \leq \boldsymbol c,\\ 
 \label{eq:p2_c2}
&\qquad\quad M \boldsymbol \phi \leq \mathbf d.
\end{align}

Observe that problems \textsf{P2} and \textsf{P1} are equivalent. The next step is to reformulate the problem \textsf{P2} into a problem with only equality-constrained as follows:

\begin{eqnarray*}
	\label{eq:p3}
	&& \textsf{P3:} \quad \min_{\boldsymbol z}  \quad f(\boldsymbol z) \nonumber\\
	&& \textrm{subject to:} \label{eq:p3_c1}  \qquad A \boldsymbol z = \boldsymbol b.\nonumber \\ 
\end{eqnarray*}

In problem \textsf{P3} we impose two additional concepts: \textit{slack variables} and \textit{logarithmic barrier function}. To express problem \textsf{P2} into equality-constrained problem we introduce two slack variables $\boldsymbol y$ and $\boldsymbol w$  associated to capacity constraint (\ref{eq:p2_c1}) and delay constraint (\ref{eq:p2_c2}), respectively. The vector $\boldsymbol y$ of length $TL$ is non-negative slack variable such that $R \boldsymbol x + \boldsymbol \sigma + \boldsymbol y = \boldsymbol c$, where $y_{tl}$ represents the slack capacity of link $l$ at time $t$. Similarly, for the constraint (\ref{eq:p2_c2}), we introduce variable $\boldsymbol w$ of length $K$ as the non-negative slack variable vector of second set of constraints, i.e., we get \mbox{$M \boldsymbol \phi + \boldsymbol w = \boldsymbol d$}, and $w_k$ represents the slack variable associated with the $k$-th delay constraint among all sources.

We introduce a decision variable $\boldsymbol z$, which is the concatenation of source rates $\boldsymbol x$, link margin vectors $\boldsymbol \sigma$, slack variable of link capacities $\boldsymbol y$, and slack variable of delay constraints $\boldsymbol w$, i.e., 
$$
\boldsymbol z = [\boldsymbol x^{\textsf{T}},\boldsymbol \sigma^{\textsf{T}}, \boldsymbol y^{\textsf{T}}, \boldsymbol w^{\textsf{T}}]^{\textsf{T}}
$$

Moreover, in problem \textsf{P3}, $\mu \geq 0$ is a coefficient for the barrier function, vector \mbox{$\boldsymbol b = [\boldsymbol c^{\textsf{T}}, \boldsymbol d^{\textsf{T}}]^{\textsf{T}}$}, and $A$ is a \mbox{$(TL+K)\times (TS+2TL+K)$} matrix defined by \mbox{$A = [F^{\textsf{T}} \quad\ G^{\textsf{T}}]^{\textsf{T}}$}, where \mbox{$F = [R \quad  I_{TL} \quad I_{TL} \quad 0_{TL \times K}]$} and \mbox{$G = [M \quad 0_{K \times 2TL} \quad I_K]$}. Finally, we get 

$$f(\boldsymbol z) =  -\sum_{t=1}^T \sum_{s=1}^S U_{st}(z_{st}) - \mu \sum_{i=1}^{TS+2TL+K} \log z_i$$

Now, the problem is equivalently  formulated in an appropriate form such that we can apply the fast distributed Newton method. 

\subsection{Distributed Newton Method}
In this subsection, we present an alternative solution based on Newton method with equality-constrained problem \cite[Ch.~10]{Boyd}. In equality-constrained Newton method the initial point must be feasible (i.e., $\boldsymbol z \in \textbf{ dom } f $ and $A \boldsymbol z = b$). Hence, at the first step, we assume that we know a feasible vector $\boldsymbol z$. For example, one such initial vector could be the minimum rate demand of each source at each time period for the rate part of $\boldsymbol z$. Indeed, this is the case that the average delay requirements are large enough; thereby the minimum rate demand of sources at all time periods are feasible points. 

The Newton algorithm produces a minimizing sequence $\boldsymbol z(j+1),\quad j=1,\dots$ given by 
$$\boldsymbol z(j+1) = \boldsymbol z(j) + \delta(j)\Delta \boldsymbol z(j),$$
where $\delta(j)$ is a positive step size and $\Delta \boldsymbol z(j)$ is the Newton direction that is given by
\begin{eqnarray}
	\label{eq:delta_z}
	\Delta \boldsymbol z(j) &=& - H^{-1}_j (\nabla f(\boldsymbol z(j))+A^{\textsf{T}} \boldsymbol \omega(j)), \\
	\label{eq:omega}
	(AH_j^{-1}A^{\textsf{T}})\boldsymbol \omega(j) &=& -AH_j^{-1}\nabla f(\boldsymbol z(j)). 
\end{eqnarray}

In Eq.~\eqref{eq:delta_z}, vector $\boldsymbol \omega(j)$ of length $TL+K$ is the dual (price) vector associated with the equality constraint of problem \textsf{P3}. The first $TL$ elements are related to the capacity constraints of all time periods and the $K$ remaining elements correspond to the average delay constraints. 
Since utility functions are strongly concave and primal vector $\boldsymbol z(j)$ is bounded, one can show that $H_j$ and $AH_j^{-1}A^{\textsf{T}}$ are both invertible \cite{NUM-Newton}. 

The iterations of Eq.~\eqref{eq:delta_z} and Eq.~\eqref{eq:omega} converge to the optimal solution of problem \textsf{P3}, which is an equivalent form of \textsf{P1}. For given values of $\boldsymbol \omega(j)$ for all links in the path of source $s$ in time periods and all of the elements associated to the delay constraints with source $s$, the value of $\Delta \boldsymbol z(j)$ can be computed locally at source $s$. However, the evaluation of inverse matrix $(AH_j^{-1}A^{\textsf{T}})^{-1}$ requires global information, and therefore, $\boldsymbol \omega(j)$ cannot be computed in a decentralized manner. In what follows, based on \cite{NUM-Newton} we present an elegant Newton method for distributed computation of dual vector $\boldsymbol \omega(j)$. 

To obtain a distributed update equation for dual variable, we present a mechanism based on matrix splitting techniques. Let us denote by $\boldsymbol \omega(j,n)$ the value of $\boldsymbol \omega$ at $n$-th dual iteration at the $j$-th primal step. For the sake of notational simplicity, we define some subsidiary functions as summarized in Table \ref{tbl:fn}. Moreover, we define \mbox{$H_j^{-1}(z_i) = [H_j^{-1}]_{ii}$} and \mbox{$\nabla f(\boldsymbol z)(z_i) = [\nabla f(\boldsymbol z)(z_i)]_i$}.

\begin{table}
	\caption{Function Definitions}
		\label{tbl:fn}
	\begin{center}
\begin{tabular}{|c|L{4.3cm}|}
			\hline
			\textbf{Function} & \textbf{(Domain and Range)/Description}\\ \hline\hline
			$l(v) = v \textrm{\textbackslash} L $ & $\{1,\dots,TL\} \rightarrow \{1,\dots,L\}$ \\ \cline{2-2}
			& The link of $v$\\
			\hline
			$t(v) = \lceil \frac{v}{L} \rceil$ & $\{1,\dots,TL\} \rightarrow \{1,\dots,T\}$ \\ \cline{2-2} 
			& The time period of $v$\\
			\hline
			$\tau(v) = \lceil \frac{v}{T} \rceil $ & $\{1,\dots,TS\} \rightarrow \{1,\dots,T\}$ \\ \cline{2-2} 
			& The time period of $v$\\
			\hline
			$s(v) = v \textrm{\textbackslash} T $ & $\{1,\dots,TS\} \rightarrow \{1,\dots,S\}$ \\ \cline{2-2}
			& The source of $v$\\
			\hline
			$p(v) = \max\limits_{s \in \mathcal{S}} \Big\{\sum\limits_{j=0}^{s-1} k_j \leq  v\Big\} $ & $\{1,\dots,K\} \rightarrow \{1,\dots,S\}$ \\ \cline{2-2} 
			& $p(v)$ is the source that $v$th delay constraint is associated to it.\\
			\hline
		\end{tabular}
	\end{center}
\end{table}

We define the weighted sum of the dual variables associated with links that are in route of source $s$ at period $t$ by 

$$ \Pi_{st}(j,n) = H_j^{-1}(x_{st})\sum\limits_{l \in \mathcal{L}} (R_t)_{ls} \omega_{l+(t-1)L}(j,n).$$

Similarly, for the delay part of constraints, we define $\Psi_{st}(j,n)$ as the weighted sum of dual variables associated with the active delay constraints of source $s$ at time $t$ by

$$\Psi_{st}(j,n) = H_j^{-1}(x_{st})\sum\limits_{k \in \mathcal{K}_s} (M_s)_{kt} \omega_{TL+j}(j,n).$$ 

Using the introduced definitions, we can rewrite the left hand side part of Eq. (\ref{eq:omega}) in the form
of Eq.~\eqref{eq:w1} and Eq.~\eqref{eq:w2}.
\begin{figure*}
\begin{align}
	\label{eq:w1}
	[AH_{j}^{-1}A^{\textsf{T}} \boldsymbol \omega(j,n)]_{v}=
	\sum_{s \in \mathcal{S}} (R_{t(v)})_{l(v)s} [\Pi_{s{t(v)}}(j,n) + \Psi_{s{t(v)}}(j,n)] + H_j^{-1}\Big(y_{t(v)l(v)}\Big)\omega_{l(v)+(t(v)-1)L}(j,n)
\end{align}
for $v = 1,\dots,TL$ and 
\begin{align}
	\label{eq:w2}
	[AH_{j}^{-1}A^{\textsf{T}} \boldsymbol \omega(j,n)]_{v} = \sum_{t \in \mathcal{T}} (M_{p(v)})_{vt}[\Pi_{p(v)t}(j,n) + 			\Psi_{p(v)t}(j,n)] + H_j^{-1}(w^v)\omega_{TL+v}(j,n)
\end{align}
for $v = 1,\dots,K$.
\hrule
\hrulefill
\end{figure*}

Using Eq.~\eqref{eq:w1} and Eq.~\eqref{eq:w2} and matrix splitting techniques presented in \cite{NUM-Newton}, now we are ready to obtain an iterative way toward distributed computation of $\boldsymbol \omega$, which is  summarized in Theorem~\ref{theorem:2}. 

\begin{myTheo}
	\label{theorem:2}
	Let $C_j$ be a diagonal matrix with \mbox{$[C_j]_{vv} = [AH_j^{-1}A^{\textsf{T}}]_{vv}$}, \mbox{$B_k = AH_j^{-1} - C_j$} be a symmetric one, $\bar B_j$ be another diagonal matrix with diagonal entries \mbox{$[\bar B_j]_{vv} = \sum\nolimits_{i=1}^{TL+M} [B_j]_{vi}$}, and the diagonal matrix $\bar C_j$ as the sum of $C_j + \bar B_j$. For each primal iteration $j$, the dual sequence $\{\omega_v(j,n)\}$ generated by the iterations Eq.~\eqref{eq:w11} and Eq.~\eqref{eq:w22} converges to the solution of Eq.~\eqref{eq:omega} as $n \rightarrow \infty$.
\begin{figure*}
\begin{align}
	\label{eq:w11}
	\omega_v(k,n+1) &= [\bar C_k^{-1}]_{vv}  \bigg([\bar B_k]_{vv} \omega_v(k,n)
					+[C_k]_{vv} \omega_v(k,n) - [AH^{-1}_kA^{\textsf{T}} \omega(k,n)]_v \nonumber\\
					&+ \sum_{s \in \mathcal{S}} (R_{t(v)})_{l(v)s} H_k^{-1}(x_{st(v)})\nabla f(\boldsymbol z(k))(x_{st(v)})+ H_k^{-1} \Big(y_{l(v)t(v)}\Big) \nabla f(\boldsymbol z(k)) \Big(y_{l(v)t(v)}\Big) \bigg)
\end{align}
for $v = 1,\dots,TL$ and 
\begin{align}
	\label{eq:w22}
	\omega_v(k,n+1) &= [\bar C_k^{-1}]_{vv}  \bigg([\bar B_k]_{vv} \omega_v(k,n)
					+ [C_k]_{vv} \omega_v(k,n) - [AH^{-1}_kA^{\textsf{T}} \omega(k,n)]_v \nonumber\\
					&+ [ C_k]_{vv} \omega_v(k,n) \sum_{t \in \mathcal{T}} (M_{p(v)})_{vt} H_k^{-1}(x_{p(v)t})\nabla f(\boldsymbol z(k))(x_{p(v)t})
					+ H_k^{-1} (z^m) \nabla f(\boldsymbol z(k)) (z^m) \bigg)
\end{align}
for $v = TL+1,\dots,TL+K$.
\hrule
\hrulefill
\end{figure*}
\end{myTheo}

Finally, in \cite{NUM-Newton} it has been shown that the update equations described in Eq.~\eqref{eq:w11} and Eq.~\eqref{eq:w22} could be calculated using merely local information at sources, thereby the corresponding algorithm is implemented in a distributed fashion. 

\section{A Solution with limited future knowledge}
\label{sec:mpc}
The solutions presented in Sections~\ref{sec:sol} and \ref{sec:newton} are based on the assumption that the problem data (input parameters) for the entire time horizon is available ahead of time. Dependence of these solutions on the precise knowledge of future network parameters stimulates devising another scheme that efficiently works under uncertainty of the parameters. In this section, we extend our solution in a way such that the problem data is not fully known in advance.  
Without loss of generality, we assume that only the link capacities are revealed at the beginning of each period. Note that this approach could be extended to capture the case that other parameters such that source utilities are not known ahead of time. Our approach in this section is based on a causality constraint such that the source rates at period $t$ is a function of the link capacities up to period $t$. 
We further note that this is a convex stochastic problem 	\cite{Bert_Stoch}, where the goal is to maximize the expected aggregated utility of all sources subject to the capacity, average delay, and causality constraints.
Like the conventional NUM problems this problem could be efficiently tackled by centralized approaches, but, here we are interested in decentralized schemes. 

\subsection{MPC-based solution}
To obtain a decentralized stochastic solution, we construct our solution based on Model Predictive Control (MPC) \cite{mpc}.
To calculate the source rates ($x_{st}$s) and link margin values ($\sigma_{tl}$s) for any particular period $\tau$, instead of solving problem \textsf{P1}, we construct and solve problem \textsf{P4} as follows

\begin{align*}
 \textsf{P4:}&  \quad  \max_{x_{st},s\in\mathcal{S},t\in\mathcal{T}^{\tau}}  \quad \sum_{t\in\mathcal{T}^{\tau}} \sum_{s\in\mathcal{S}} U_{st}(x_{st})\\
	&\textrm{subject to:} \\
	&\qquad\quad R_\tau X\mathbf e_\tau +\boldsymbol \sigma_{\tau}\leq \mathbf c_{\tau},\\ 
			&\qquad\quad R_t X\mathbf e_t +\boldsymbol \sigma_{t}\leq \hat{\mathbf  c}(t|\tau), \qquad\forall  t \in \mathcal{T}^{\tau},\\ 
&\qquad\quad M_{s} \boldsymbol \phi_{s} \leq \mathbf d_s, \qquad\qquad\forall s \in \mathcal{S},\\
&\qquad\quad\phi_{st} = \sum_{l\in\mathcal L} (R_t)_{ls}D(\sigma_{tl}), \quad\forall s \in \mathcal{S}, \forall t\in\mathcal T, 
\end{align*}
where $\mathcal{T}^{\tau} = \{\tau+1,\dots,T\}$ and ${\hat{\mathbf c}(t|\tau) = E[\boldsymbol c_t|\boldsymbol c_1,\dots,\boldsymbol c_{\tau}],t\in\mathcal{T}^{\tau}}$ is the expected value of link capacities, given the entire information at period $\tau$. Consequently, in problem \textsf{P4}, at any period $\tau$ the whole information about link capacities before and at period $\tau$ is revealed. Furthermore, for the future periods we use the conditional mean values of link capacities. In addition, since each source can declare several average delay constraints, it is conceivable that some delay constraints are expired before beginning of period $\tau$, i.e., the active interval of the constraint is started and finished before period $\tau$. Another situation is when a contract has already been active. That is, the start time is ``$\leq \tau$'', while the final time is ``$\geq \tau$''. Here, the source rate for periods $\leq \tau$ are already calculated, let denote them by $x'_{st}, s\in\mathcal{S},t\in \{1,\dots,\tau\}$.
Then, the contract inequality is interpreted as follows: the source rates for $t < \tau$ are taken to be $x'_{st}$. And so, this part is fixed and the constraint should be satisfied for the remaining part in period $t \geq \tau$. 

Finally, problem \textsf{P4} is a particular version of problem \textsf{P1} and both the optimal algorithms in Sections~\ref{sec:sol} and \ref{sec:newton} could be employed to find its optimal solution. But, in each period $\tau$, we pick the optimal source rate and link margin values for just the time slot $\tau$ and then we solve the problem again for the remaining time slots. For each time slot, the optimal values of the previous periods are used as input parameters. 

\section{Experimental Results}
This section is devoted to the experimental results. First, we concentrate on a tractable network topology to verify the correctness of DA-DNUM. Second, by describing two comparison scenarios, we investigate both the superiority of our work against similar approaches as well as scalability of DA-DNUM.
\label{sec:sim}
\subsection{Experiment 1: Simple and Tractable Topology}
In order to facilitate the detail discussion of the results, we have chosen a network with time-invariant routing and topology shown in Fig. \ref{fig:top}. We set $T=10$ and $c_{1t}$ and $c_{4t}$ are chosen uniformly at random from $[4,6]$, and $c_{2t}$ and $c_{3t}$ are randomly and uniformly drawn from $[4,10]$. We choose $U_{st}(x_{st}) = \log x_{st}$ for all $s$ and $t$. Also, we assume that $D(z)={1 \over z}$ for all links that represents M/M/1 queuing model. 
We give delay indicator matrices as well as vectors $\mathbf d_s,s\in\mathcal S$ below:

\begin{eqnarray}
	M_1 &=& {1 \over 3}\times \left[\begin{matrix}
		1 & 1 & 1 & 0 & 0 & 0 & 0 & 0 & 0 & 0\\
		0 & 0 & 0 & 0 & 0 & 1 & 1 & 1 & 0 & 0
	\end{matrix}\right],\nonumber\\ 
	M_2 &=& {1 \over 6} \times \left[\begin{matrix}
		1 & 1 & 1 & 1 & 1 & 1 & 0 & 0 & 0 & 0
	\end{matrix}\right], \nonumber\\
	M_3 &=& {1 \over 6} \times \left[\begin{matrix}
		0 & 0 & 1 & 1 & 1 & 1 & 1 & 1 & 0 & 0
	\end{matrix}\right], \nonumber\\
	M_4 &=& {1 \over 4} \times \left[\begin{matrix}
		0 & 0 & 1 & 1 & 1 & 1 & 0 & 0 & 0 & 0
	\end{matrix}\right], \nonumber\\
\mathbf d_1 &=& \left[2\quad 1\right]^{\textsf T},\qquad d_2=d_3= 2,\qquad d_4 = 2.5\nonumber.
\end{eqnarray}

 \begin{figure}[b]
 \begin{center}
 \includegraphics[angle=0,scale=.35]{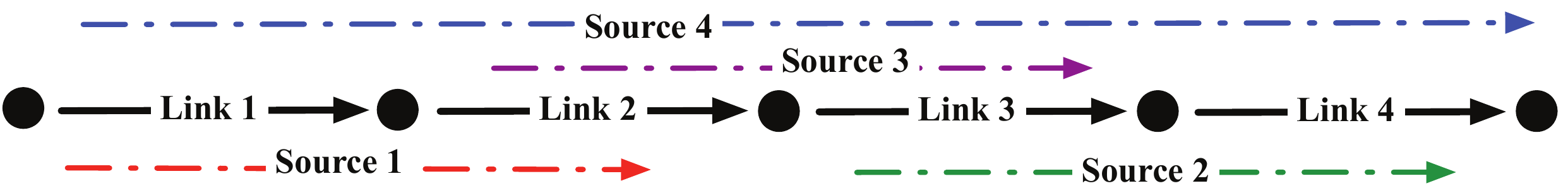}
 \end{center}
 \caption{Network Topology, Experiment 1}
 \label{fig:top}
 \end{figure}
%

 \begin{figure*}[!t]
 \centering
 \subfigure[Optimal Source Rates]{
 \includegraphics[angle=0,scale=.32]{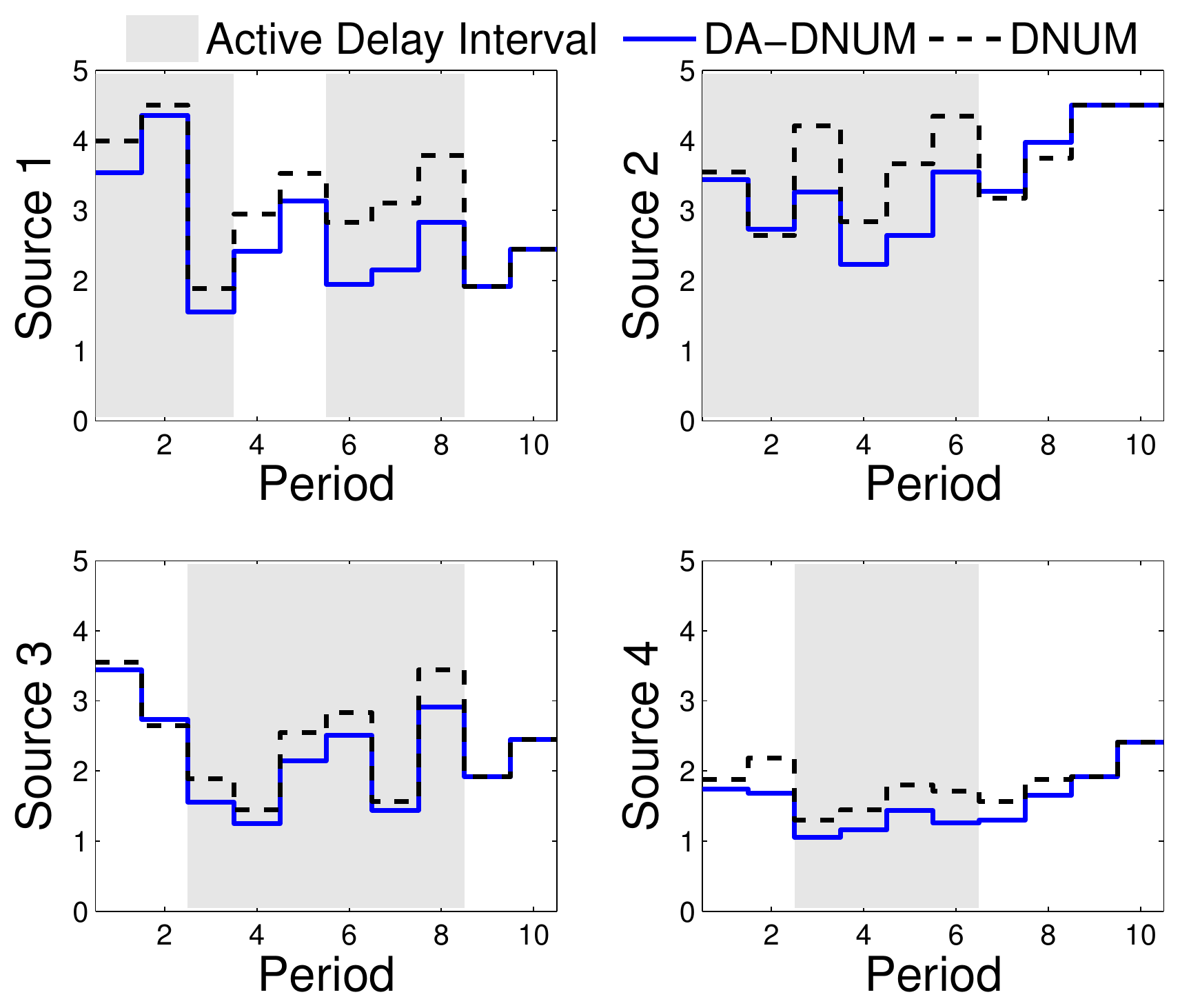}
 \label{fig:rate}}
  \subfigure[Source Delays]{
 \includegraphics[angle=0,scale=.32]{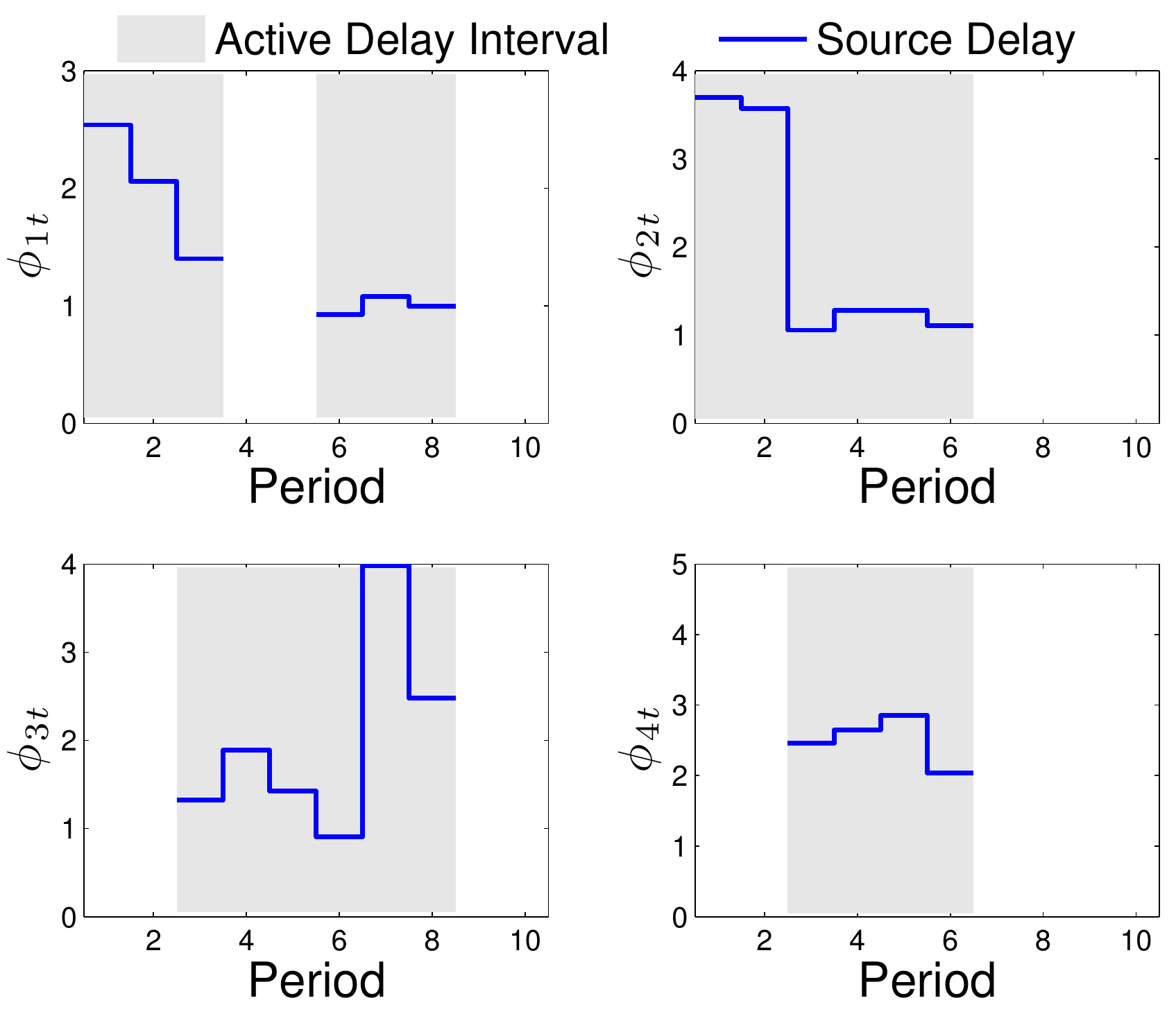}
  \label{fig:delay}}
 \subfigure[Link Capacities]{
 \includegraphics[angle=0,scale=.32]{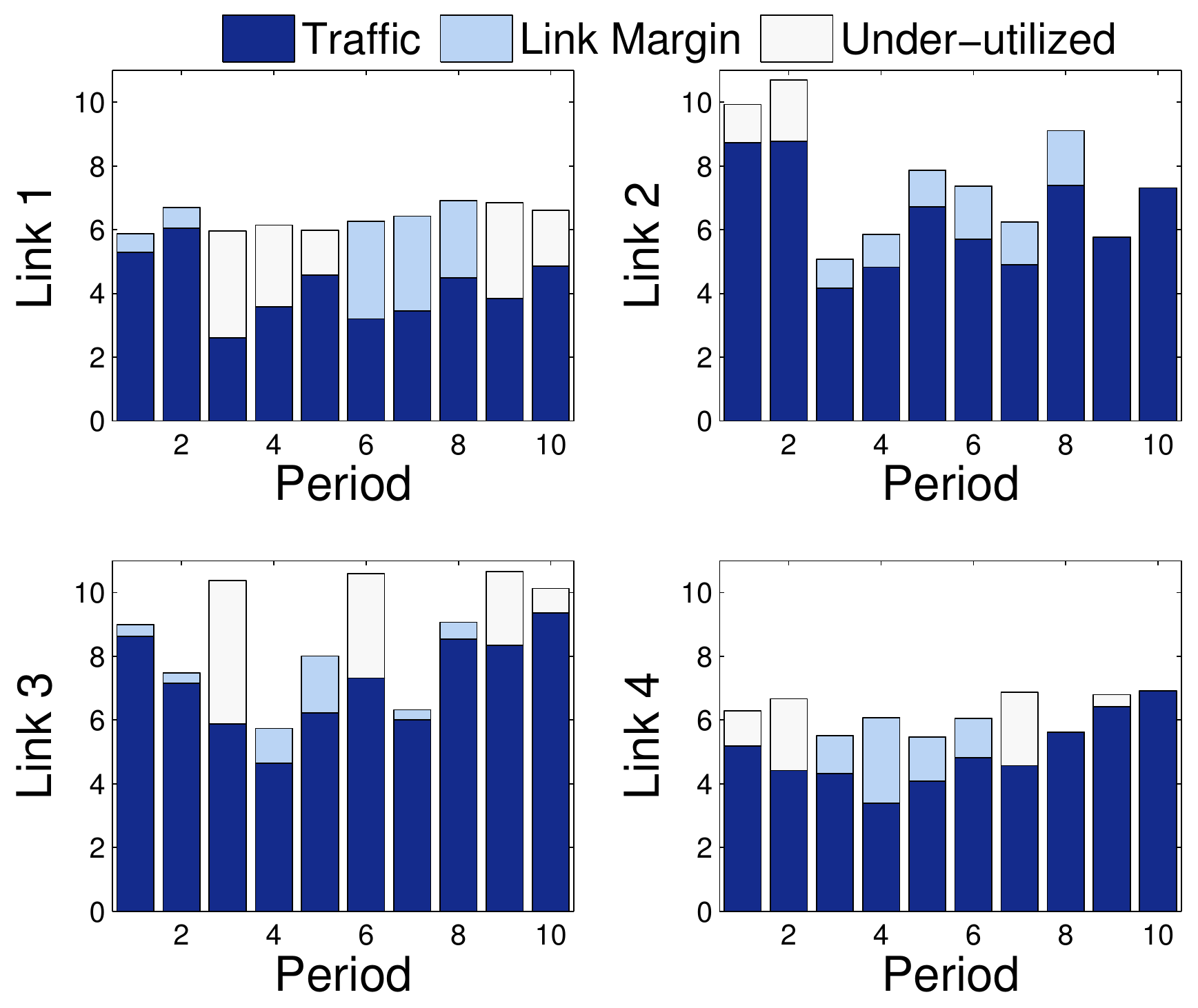}
 \label{fig:link}}
 \caption[]{Results of Experiment 1}
 \label{fig:res}
 \end{figure*}

We stress that the above delay indicator matrices imply that for $t=9,10$, there is no delay constraint and thus for these periods, \textsf{P1} degenerates to DNUM \cite{DNUM} without delivery contracts.

Fig. \ref{fig:res} displays the rate allocation result obtained from DA-DNUM algorithm with $\gamma = 0.01$ and $\texttt{th}=0.01$. For the sake of comparison, Fig. \ref{fig:res} also shows the rate allocation result of DNUM (without delivery contracts), which is obtained by solving  \textsf{P1} after removal of delay constraints. 
Fig. \ref{fig:rate} shows final source rates of the two cases. As we expect, Fig. \ref{fig:rate} exhibits  the same values for both DA-DNUM and DNUM for $t=9,10$. By contrast, for $t=1,\dots,8$ source rates obtained by DA-DNUM are lower than those provided by DNUM. This stems from existence of at least one delay constraint in any of these periods. 

End-to-end queuing delays $\phi_{st}$ for all $s$ and $t$ are depicted in Fig. \ref{fig:delay}.
To achieve higher system-wide aggregate utility, DA-DNUM allows some fluctuations in source delays during periods, while the average delays do not exceed $\mathbf d_s$. 
This flexibility in rate allocation due to the time-varying algorithm design can yield a wider feasible rate allocation schemes in comparison with the single-period NUM that is expressed in the next subsection by another experiment. 
Finally, Fig. \ref{fig:link} shows link traffics, link margins, and the amount of under-utilized link capacities. Clearly, in periods $t=9,10$, all links possess zero link margins. On the other hand, for $t=1,\dots,8$, positive values for link margin variables (for at least one link) evince that there is at least one active delay constraint imposed by the sources.   

We also have executed another experiment with the same settings of the previous one to verify the correctness and analyze the convergence behavior of distributed Newton solution that is proposed in Section~\ref{sec:newton}. As expected, the optimal rate and link margin values are the same for DA-DNUM and distributed Newton method, since both solutions are optimal. However, DA-DNUM converges to the optimal values nearly after 320 iterations with $\texttt{th}=0.01$. This value is about 34 iterations for the distributed Newton method, which reveals a significantly faster convergence rate (in terms of number of iterations) compared to the DA-DNUM. 

\begin{figure}[b]
\begin{center}
\includegraphics[angle=0,scale=.27]{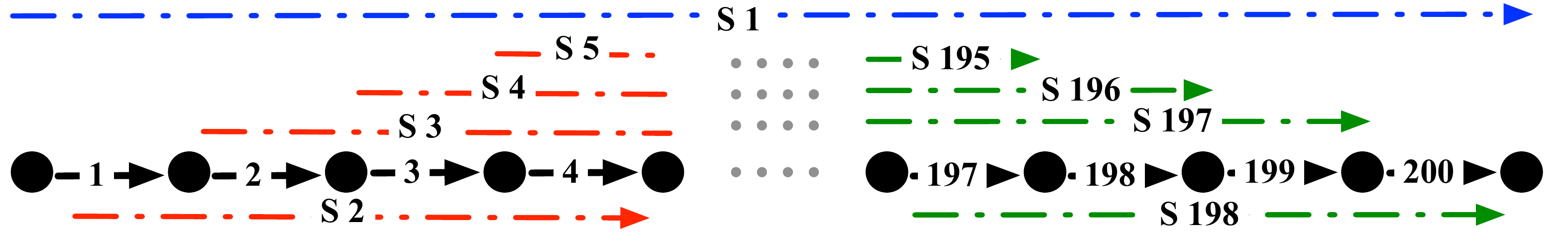}
\end{center}
\caption{Network Topology, Experiment 2}
\label{fig:top2}
\end{figure}


We also investigate the effectiveness of the MPC-based solution that is proposed in Section~\ref{sec:mpc}. Toward this, we used the same topology and settings as the previous one. But, in each period we should decide for the expected link capacity values.  For all $\tau$s, we use $\hat{c}_1(t|\tau) = \hat{c}_4(t|\tau) = 5$ and $\hat{c}_2(t|\tau) = \hat{c}_3(t|\tau) = 7$.
We compare the results of problem \textsf{P4} with the knowledge of the current and the previous periods with the global problem \textsf{P1} with the entire knowledge of the horizon. The resulting source rates and link margin values are quite similar. To illustrate the accuracy of the suboptimal MPC-based solution we report the utility values obtained by two solutions. The aggregated utility value for MPC-based solution is 36.16; while the utility value obtained by the DA-DNUM solution (as the solution of problem \textsf{P1}) is 37.01 which evinces 2.2\% difference. 

\subsection{Experiment 2: Comparison Scenario} We next compare DA-DNUM  with the algorithm proposed in \cite{QiuDelay} (by assuming fixed capacities) in a large-scale scenario. We remark that the algorithm proposed in \cite{QiuDelay} is based on the single-period version of NUM that is customized in delay-sensitive setting. Consequently, Single-period NUM in algorithm of \cite{QiuDelay} persuades us to solve $T$ separate problems for the entire $\mathcal T$.  We consider a line topology with $200$ links and $198$ sources (Fig. \ref{fig:top2}) whose $200 \times 198$ routing matrix is given in below:

\begin{equation}
	\label{eq:R}
 R_t = \left[  \begin{matrix} 1 & 1 & 0 & 0 &\dots & 0 & 0 \\
		1 & 1 & 1 & 0 & \dots & 0 & 0\\
		1 & 1 & 1 & 1 & \dots & 0 & 0\\
		1 & 1 & 1 & 1 & \dots & 0 & 0\\
		1 & 0 & 1 & 1 & \dots & 0 & 0\\
		\vdots & \vdots & \vdots & \vdots & \ddots & \vdots & \vdots\\ 
		1 & 0 & 0 & 0 & \dots & 0 & 1\end{matrix}  \right].
\end{equation}
\normalsize

In addition, the other parameters are listed in Table I. To clearly exhibit the different behavior of DA-DNUM, we intentionally set up only 2 average delay constraints for source 1 (the source with all links on its path) and source 2 (the one that traverses through first 4 links).

To exhibit the flexibility of DA-DNUM, this experiment simply obliges a minimum rate demand as ${x_{1,2}^{\min} = 5}$. This means that the minimum rate requirement of source 1 at period 2 is 5.
The aforementioned minimum rate demand is in conflict with the average delay requirement since the higher rate results in higher end-to-end delay according to the limited capacity of links. Nonetheless,  DA-DNUM easily remedies this conflicting situation by assigning the declared minimum rate to $s_1$ at $t_2$, thus enduring a larger short-term delay (around 85 instead of $\mathbf d_1 = 50$). Thanks to supporting time-coupled system model, DA-DNUM allocates proper rates to this source in other periods, so as to maintain the average delay below 50. In contrast, the single-period algorithm of \cite{QiuDelay} fails for this scenario since the underlying NUM becomes infeasible. This simple experiment signifies the relatively wider set of feasible problems of DA-DNUM. One may construct several other feasible scenarios for DA-DNUM that are infeasible for the problem of \cite{QiuDelay}.

\begin{table}[]
\footnotesize
\caption{Parameters of Experiment 2}
\label{tbl_pr}
\begin{center}
\begin{tabular}{|c|c|}
\hline
Parameter & Value\\
\hline\hline
$S$ & 198\\
\hline
$L$ & 200\\
\hline
$T$ & 50\\
\hline
$c_{tl}, t \in \mathcal{T},l\in \mathcal{L}$ & [8,12] \emph{kbps}\\
\hline
$k_s, s \in \{1,2\}$ & 1\\
\hline
$M_s, s \in \{1,2\}$ & $[1/50]_{1 \times 50}$\\
\hline 
$k_s, s \in \{3,\dots,198\}$ & 0 \\
\hline
$M_s, s \in \{3,\dots,198\}$ & $[0]_{1 \times 50}$ \\
\hline
$\mathbf d_s, s \in \{1,2\} $ & $50$\\
\hline

\end{tabular}
\end{center}
\normalsize
\end{table}

%

 \begin{table}[]
 \caption{Parameters of Experiments 3}
 \footnotesize
 \label{tbl_pr}
 \begin{center}
 \begin{tabular}{|c|c|}
 \hline
 \textbf{Parameter} & \textbf{Value}\\
 \hline\hline
  $S$ & 20\\
 \hline
  $L$ & 20\\
 \hline
  $T$ & 20\\
 \hline
 $c_{tl}, t \in \mathcal{T},l\in \mathcal{L}$ & 20 \emph{kbps} \\
 \hline
 $k_s, \quad s \in \mathcal{S}$ & 1\\
 \hline 
 
 $M_s,\quad s \in \mathcal{S}$ & \emph{random} \\
 \hline
  $\boldsymbol d_s, \quad s \in \mathcal{S} $ & $[4,6]$\\
 \hline
 
 \end{tabular}
 \end{center}
 \end{table}
 \normalsize

\subsection{Experiment 3: Random Topology}
We examine DA-DNUM for the case of more complex randomly generated topologies. We run DA-DNUM and algorithm of \cite{QiuDelay} for a scenario with a random topology comprising $20$ sources and $20$ links (see Table I for the parameters). As it allows temporal fluctuations in source delays, DA-DNUM yields slightly better link utilization  compared to \cite{QiuDelay}: The under utilized link capacity averaged over all links and all periods for the algorithm of \cite{QiuDelay} is 4.06 whereas it is 3.91 for DA-DNUM. Hence, $3.7\%$ improvement is obtained. Consequently, by these two experiments, we show both wider range of feasibility along with better resource utilization of DA-DNUM against existing single-period approaches.

\section{Conclusion and Future Directions}
\label{sec:conc}
To ameliorate QoS experience in real-time networking applications in terms of guaranteeing fixed average end-to-end delay over long periods, we addressed a dynamic NUM problem with source-driven time-coupled constraints on average end-to-end delay. We proposed two set of solutions, \textit{first}, a DA-DNUM algorithm as the dual-based distributed solution of the formulated optimization problem. \textit{Second,} we devised another solution based on the recently-proposed distributed Newton method to improve the slow convergence rate of the DA-DNUM. DA-DNUM allocates source rates in a way that achieves the maximum network-wide utility aggregated over time interval while satisfying capacity and delay constraints. Numerical experiments exhibited that, compared to existing schemes, DA-DNUM admits relatively wider feasible scenarios along with higher resource utilization. This enhancement originated from multi-period problem setup that allows short-term delay fluctuations while keeps long-term value around the required one.
Obtained results stimulate further research activities. A promising line is to investigate the solution when link delay function is a non-convex function which is not far from the reality in the case of complicated packet arrival models. In addition, we plan to extend this work in wireless scenarios to jointly consider the delay-aware rate allocation and link scheduling as a cross-layer solution design. 

\bibliography{ref}
\bibliographystyle{IEEEtran}


\end{document}